\documentclass[journal]{IEEEtran}
\ifCLASSINFOpdf

\else

\fi
\usepackage{amsmath}
\usepackage{makeidx}  
\usepackage{algorithm}
\usepackage{algorithmic}
\usepackage{graphicx}
\usepackage{subfigure}
\usepackage{epstopdf}
\usepackage{bm}
\usepackage{cite}
\usepackage{stfloats}

\usepackage{amssymb}
\usepackage{graphicx}

\usepackage{url}

\usepackage{tabularx,booktabs}
\newcolumntype{C}{>{\centering\arraybackslash}X} 
\usepackage{lipsum}

\usepackage{makecell} 

\usepackage{graphicx}
\usepackage{color}
\newtheorem{thm}{Theorem}
\newtheorem{rem}{Remark}
\newtheorem{lem}{Lemma}

\begin{document}

\title{Achievable Rate Region Characterization of Rotatable Antenna-assisted Interference Channel}

\author{Junran Pei, Guangji Chen,
         Qingqing Wu, Qiaoyan Peng, and Wei Cai \vspace{-22pt}
        \thanks{Junran Pei and Guangji Chen are with Nanjing University of Science and Technology, Nanjing 210094, China (email: junranpei@njust.edu.cn; guangjichen@njust.edu.cn). Qingqing Wu is with Shanghai Jiao Tong University, 200240, China (e-mail: qingqingwu@sjtu.edu.cn).
        Qiaoyan Peng is with University of Macau, Macao 999078, China (email: yc27464@umac.mo).
        Wei Cai is with School of Computer Engineering, Nanjing Institute of Technology, Nanjing 211167, China (email: caiw@njit.edu.cn).}}

\maketitle
\vspace{-3pt}
\begin{abstract}
	This paper characterizes the achievable rate region of a rotatable antenna (RA)-assisted multiple-input single-output (MISO) interference channel.
	We exploit the additional spatial degrees of freedom (DoFs) offered by the dynamic orientation of RAs to simultaneously enhance the desired signal and suppress interference, thereby significantly enlarging the achievable rate region of the interference network.
	To this end, we invoke the rate-profile technique to formulate a sum-rate maximization problem for jointly optimizing the transmit beamforming and the array rotation, subject to the practical mechanical constraints of the RAs.
	For the special case of a two-receiver scenario, we introduce the concept of subspace correlation coefficient to derive closed-form solutions for the optimal RA rotation angles and transmit beamforming vectors.
	Inspired by the insights obtained in the two-receiver special case, we propose an efficient two-step optimization framework to obtain a high-quality solution for the general multi-receiver case. In particular, the array rotation angle of each transmitter is optimized to minimize the leakage interference to non-associated receivers. With the optimized array rotation angles, the optimal transmit beamforming is obtained by invoking the second order cone programming based approach.
	Numerical results demonstrate that the proposed solution achieves near-optimal performance and that integrating RAs effectively simplifies transceiver designs by enabling simple maximum ratio transmission (MRT) beamforming with negligible performance loss.
\end{abstract}


\IEEEpeerreviewmaketitle

\section{Introduction}
\vspace{-2pt}
In current wireless networks, the ever-increasing node density and universal frequency reuse have made co-channel interference management a critical challenge.
The interference channel (IFC), which models a typical scenario where multiple transmitters communicate with their respective receivers simultaneously over the same spectrum, is one of the most classic models in wireless communication systems from an information-theoretic perspective.
For practical and low-complexity implementation, a widely adopted approach in the multiple-input single-output (MISO) IFC is to treat interference as noise at the receivers while employing coordinated transmit beamforming at the transmitters \cite{timotheou2014beamforming}.

\begin{figure}[!t]
	\centering
	\includegraphics[width=0.4\textwidth]{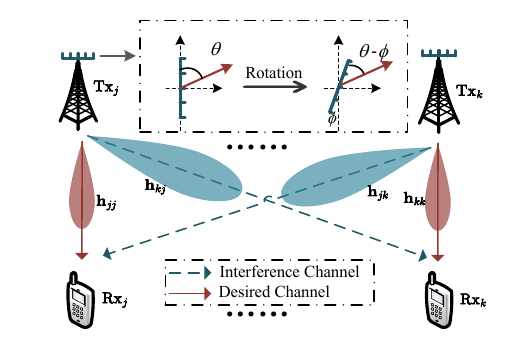}
	\DeclareGraphicsExtensions.
	\caption{System model of the RA-enabled MISO IFC.}
	\label{model}
	\vspace{-16pt}
\end{figure}

However, conventional MISO interference networks predominantly rely on fixed-position antennas (FPAs), where the static spatial deployment may not fully unlock the spatial degrees of freedom (DoFs) due to the randomness of wireless channels. To overcome this inherent limitation, rotatable antenna (RA) has emerged as a new type of flexible antenna system that dynamically adjusts the array orientation to proactively reshape wireless channels \cite{zheng2026rotatable, shao2026tutorial,dai2025demo,dai2026rotatable,zheng2026rotatable_comm,dai2026covert,dai2025secure}. To fully unleash the potential of RA for channel reconfigurations, the joint optimization of transmit beamforming and antenna orientation has attracted growing interest for improving wireless network performance. Recently, RAs have been actively explored in diverse application scenarios, including integrated sensing and communications (ISAC) \cite{chen2026rotatable}, spectrum sharing \cite{peng2026rotatable}, and MISO multiple access channel \cite{wang2024movable}. To suppress co-channel interference in a spectrum sharing system, the authors of \cite{tan2026rotatable} investigated the joint design of RA rotation and transmit beamforming, but the model is limited to a two-receiver topology with a single secondary receiver. Extending to a general multi-receiver setup, the authors of \cite{peng2026rotatable} studied the weighted sum-rate maximization in an RA-assisted spectrum sharing scenario. Although these works \cite{tan2026rotatable, peng2026rotatable} demonstrate the benefits of RAs in terms of interference suppression, the performance gap between the achievable performance and system performance limit remains unclear. It is well known that the achievable rate region is widely adopted to evaluate the performance of IFC from an information theoretic viewpoint \cite{huang2020achievable}.  The achievable rate region of an RA-assisted MISO IFC remains largely unknown, which motivates our work. 

This paper investigates an RA-assisted MISO IFC to characterize its achievable rate region. To gain useful insights into the system design, we first focus on a special case of a two-receiver setup. By introducing the subspace channel correlation coefficient, we analytically derive the condition for achieving the zero subspace channel coefficient, thereby unveiling the potential of RAs in achieving the interference-free upper bound. The subsequent optimal solutions of RA rotation and transmit beamforming are derived in closed-form expressions. Building on these insights, we extend the solution to general multi-receiver scenarios by proposing a two-stage optimization algorithm with low complexity. In this general setup, the rotation of RA at each transmitter is optimized to minimize the leakage interference to non-associated receivers. Given the obtained rotation angles, the optimal transmit beamforming is obtained via the second-order cone programming (SOCP) approach. Simulation results validate the theoretical findings for the two-receiver special case and demonstrate the superiority of the proposed RA-assisted MISO IFC in terms of the system capacity. Moreover, simulation results demonstrate simple maximum ratio transmission (MRT) beamforming is able to approach the interference-free upper bound under the favorable channel conditions created by RAs, which provide system design guidelines for simplifying the transceiver design with negligible performance loss.

\section{System model}
As shown in Fig. 1, we consider an RA-assisted MISO IFC with $K$ transmitter-receiver pairs, where each transmitter equipped with a uniform linear RA array of $N$ transmit antennas simultaneously transmits signals to its associated single-antenna receiver. Without loss of generality, we adopt a two-dimensional (2D) Cartesian coordinate system, where transmitter $j$ ($j \in \mathcal{J} \triangleq\{1, \dots, K\}$) and receiver $k$ ($k \in\mathcal{K} \triangleq \{1, \dots, K\}$) are located at $\mathbf{q}_{\mathrm{tx},j} = [x_{\mathrm{tx},j}, y_{\mathrm{tx},j}]^T$ and $\mathbf{q}_{\mathrm{rx},k} = [x_{\mathrm{rx},k}, y_{\mathrm{rx},k}]^T$, respectively. The distance from transmitter $j$ to receiver $k$ is $d_{kj} = \|\mathbf{q}_{\mathrm{tx},j} - \mathbf{q}_{\mathrm{rx},k}\|$. The transmit steering vector is given by $\mathbf{a}_t(\theta) = \left[1, e^{i \pi \sin \theta}, \dots, e^{i \pi (N-1) \sin \theta}\right]^T$. Assuming all channels are dominated by line-of-sight (LoS) components\footnote{This assumption is justified in unmanned aerial vehicle device-to-device communication networks. In such scenarios, since transceivers operate at certain altitudes, the LoS path dominates while multipath scattering is negligible.}, the wireless channel from transmitter $j$ to receiver $k$ is given by
\begin{align}\label{channel_model}
	\mathbf{h}_{kj}(\phi_j) = \sqrt{\beta_0} \frac{e^{-i \frac{2\pi}{\lambda} d_{kj}}}{d_{kj}} \mathbf{a}_{\mathrm{t}}(\theta_{kj} - \phi_j),
\end{align}
where $\beta_0$ represents the channel power gain at a reference distance of 1 meter (m), $\lambda$ is the carrier wavelength. $\theta_{kj}$ denotes the spatial angle between transmitter $j$ and receiver $k$, and $\phi_j$ denotes the rotation angle of the antenna array at transmitter $j$.

Let $\mathbf{w}_j\in\mathbb{C}^{N\times1}$ denote the transmit beamforming vector of transmitter $j$. The transmit power of each transmitter is constrained by $\|\mathbf{w}_j\|^2 \le P_{\max}$, where $P_{\max}$ denotes the maximum allowed transmit power. The baseband signal received at receiver $k$ can be expressed as
\begin{align}\label{received_signal}
	y_k &= \sum_{j=1}^{K}\mathbf{h}_{kj}^H(\phi_j)\mathbf{w}_js_j+z_k \nonumber \
	\\&= \mathbf{h}_{kk}^H(\phi_k)\mathbf{w}_ks_k+\sum_{j\neq k}\mathbf{h}_{kj}^H(\phi_j)\mathbf{w}_j\,s_j+z_k, \forall k \in \mathcal{K},
\end{align}
where $s_k$ denotes the Gaussian information-bearing symbol for receiver $k$ with normalized power, $z_k$ is the additive white Gaussian noise (AWGN) at receiver $k$ with $z_{k}\sim\mathcal{C N}\left(0,\sigma_{k}^{2}\right)$. By treating interference as noise, the resulting signal-to-interference-plus-noise ratio (SINR) at receiver $k$ is given by
\begin{equation}
	\gamma_k=\frac{|\mathbf{h}_{kk}^H(\phi_k)\mathbf{w}_k|^2}{\sum_{j\neq k}|\mathbf{h}_{kj}^H(\phi_j)\mathbf{w}_j|^2+\sigma_k^2},\forall k \in \mathcal{K}.
\end{equation}
Accordingly, the achievable rate of receiver $k$ is
\begin{align}\label{Rate_k_definition}
	R_k &= \log_2(1 + \gamma_k) \nonumber \\
	&= \log_2 \left( 1 + \frac{|\mathbf{h}_{kk}^H(\phi_k) \mathbf{w}_k|^2}{\sum_{j \neq k} |\mathbf{h}_{kj}^H(\phi_j) \mathbf{w}_j|^2 + \sigma_k^2} \right), \forall k \in \mathcal{K}.
\end{align}

The achievable rate region for the IFC is the set of all rate-tuples $(R_{1}, ..., R_{k})$ for the $K$ transmitter-receiver pairs that can be simultaneously achieved, which can be written as
\begin{align}\label{Rate_Region}
\mathcal{R} = \bigcup_{\substack{\|\mathbf{w}_j\|^2 \le P_{\max} ,0 \le \phi_j < \pi \\ j \in  \mathcal{J}}} \{ (R_1,R_2, \dots, R_K) \}.
\end{align}
The outer boundary of $\mathcal{R}$ is called the Pareto boundary, which consists of all the rate-tuples for which it is impossible to increase one receiver's rate without simultaneously decreasing that of any other receiver \cite{huang2020achievable}. To characterize the Pareto-optimal rate tuples of the considered RA assisted MISO IFC, we invoke the technique of rate-profile. Specifically, any rate-tuple on the Pareto boundary of $\mathcal{R}$ can be obtained by optimally solving the following optimization problem with a given rate-profile ratio weight vector $\boldsymbol{\zeta}=(\zeta_1,\dots,\zeta_K)$:
\begin{subequations}\label{P1}
\begin{align}
\label{P1-a}\max_{R,{\{\phi_j\}}_{j=1}^K,{\{\mathbf{w}_j\}}_{j = 1}^K} \;\; & R \\
\label{P1-b}\text{s.t.} \;\;\;\;\;\;\;\;\;\;\;\; &R_k \ge \zeta_k R, \forall k \in \mathcal{K}, \\
\label{P1-c}&  \|\mathbf{w}_j\|^2 \le P_{\max}, \forall j \in \mathcal{J}, \\
\label{P1-d}& 0 \le \phi_j < \pi, \forall j \in \mathcal{J},
\end{align}
\end{subequations}
where $\zeta_k\ge0$ denotes the target rate ratio between the achievable rate of receiver $k$ and the sum rate $R$, with $\sum_{k=1}^K \zeta_k = 1$. By solving problem \eqref{P1} with all possible values of $\boldsymbol{\zeta}$, the complete Pareto boundary of the achievable rate region $\mathcal{R}$ can be characterized. For problem \eqref{P1}, constraint \eqref{P1-b} specifies the rate-profile requirement, ensuring that the achievable rate of each receiver $k$ meets the target rate ratio $\zeta_k$ relative to $R$. Constraint \eqref{P1-c} imposes the per-transmitter power constraint, ensuring that the transmit power of each transmitter $j$ does not exceed the maximum power budget $P_{\max}$. Constraint \eqref{P1-d} restricts the antenna rotation angle $\phi_j$ at each transmitter to $[0, \pi)$.

Different from conventional MISO IFC, our formulation involves a new design variable, namely the antenna rotation angles $\{\phi_j\}$. This inherently offers a new dimension to enhance the achievable rate by deliberately shaping the spatial interference subspace. However, directly solving the joint optimization of beamforming and rotation is highly challenging. In particular, the antenna rotation angle affects the channel vectors in a complicated, non-linear manner. Hence, there is no standard method to solve problem \eqref{P1} optimally.

\section{Proposed Solution}

In the following, we first focus on a special case of the two-receiver scenario, where the optimal solution is derived in closed form, providing useful insights for the general setup. Inspired by these insights, we propose an efficient algorithm to obtain a high-quality solution of problem \eqref{P1} for the general multi-receiver scenario.

\subsection{Special Case of Two-receiver Setup}
For the two-receiver special case, we derive the optimal solution to unveil the additional DoFs offered by RAs for achieving the system interference-free upper bound.
To this end, we first establish an upper bound for the achievable rates.

\begin{lem}\label{lemma_upper_bound}
	The achievable rate of receiver $k$ satisfies
	\begin{align}\label{eq_upper_bound}
		R_k \le \log_2 \left( 1 + \frac{P_{\max} N \beta_0}{d_{kk}^2 \sigma_k^2} \right), \quad \forall k \in \{1, 2\},
	\end{align}
	where the equality holds if and only if $\mathbf{h}_{kk}^H(\phi_k) \mathbf{h}_{kj}(\phi_j) = 0$ for $j \neq k$.
\end{lem}
\begin{IEEEproof}
	From \eqref{Rate_k_definition}, $R_k$ is bounded by
	\begin{align}
		R_k &= \log_2 \left( 1 + \frac{|\mathbf{h}_{kk}^H(\phi_k) \mathbf{w}_k|^2}{\sum_{j \neq k} |\mathbf{h}_{kj}^H(\phi_j) \mathbf{w}_j|^2 + \sigma_k^2} \right) \nonumber \\
		&\le \log_2 \left( 1 + \frac{|\mathbf{h}_{kk}^H(\phi_k) \mathbf{w}_k|^2}{\sigma_k^2} \right) \label{eq_proof_ineq1} \\
		&\le \log_2 \left( 1 + \frac{P_{\max} \|\mathbf{h}_{kk}(\phi_k)\|^2}{\sigma_k^2} \right) \label{eq_proof_ineq2} \\
		&= \log_2 \left( 1 + \frac{P_{\max} N \beta_0}{d_{kk}^2 \sigma_k^2} \right). \label{eq_proof_eq3}
	\end{align}
	The equality in \eqref{eq_proof_ineq1} holds if and only if $\sum_{j \neq k} |\mathbf{h}_{kj}^H(\phi_j) \mathbf{w}_j|^2 = 0$. The equality in \eqref{eq_proof_ineq2} holds if and only if the transmit beamforming is set to $\mathbf{w}_k = \sqrt{P_{\max}} \frac{\mathbf{h}_{kk}(\phi_k)}{\|\mathbf{h}_{kk}(\phi_k)\|}$. The equality in \eqref{eq_proof_eq3} holds because $\|\mathbf{h}_{kk}(\phi_k)\|^2 = \frac{N \beta_0}{d_{kk}^2}$, which is determined solely by the path loss and array size, and is independent of the rotation angle $\phi_k$. 
\end{IEEEproof}

Lemma \ref{lemma_upper_bound} indicates that achieving the rate upper bound requires complete spatial orthogonality between the desired channel and the interference channel. To quantitatively measure this condition, we define the subspace correlation coefficient as $\rho_{kj}^2(\phi_j) = \frac{|\mathbf{h}_{kj}^H(\phi_j) \mathbf{h}_{jj}(\phi_j)|^2}{\|\mathbf{h}_{kj}(\phi_j)\|^2 \|\mathbf{h}_{jj}(\phi_j)\|^2}$. Then, we provide a closed-form expression of $\rho_{kj}^2(\phi_j)$ in the following lemma.

\begin{lem}\label{lemma_correlation}
	With $\Delta_{kj} \triangleq \sin(\theta_{kj} - \phi_j) - \sin(\theta_{jj} - \phi_j)$, the correlation coefficient $\rho_{kj}^2(\phi_j) \in [0,1]$ is given by
	\begin{align}\label{eq_correlation}
		\rho_{kj}^2(\phi_j) = \frac{1}{N^2} \left| \frac{\sin\left( \frac{N \pi}{2} \Delta_{kj} \right)}{\sin\left( \frac{\pi}{2} \Delta_{kj} \right)} \right|^2.
	\end{align}
\end{lem}
\begin{IEEEproof}
	From \eqref{channel_model}, $\rho_{kj}^2(\phi_j)$ can be expressed as
	\begin{align}
		\rho_{kj}^2(\phi_j) &= \frac{|\mathbf{h}_{kj}^H(\phi_j) \mathbf{h}_{jj}(\phi_j)|^2}{\|\mathbf{h}_{kj}(\phi_j)\|^2 \|\mathbf{h}_{jj}(\phi_j)\|^2} \nonumber \\
		&= \frac{\left| \frac{\beta_0}{d_{kj}d_{jj}} \mathbf{a}_{\mathrm{t}}^H(\theta_{kj} - \phi_j) \mathbf{a}_{\mathrm{t}}(\theta_{jj} - \phi_j) \right|^2}{\frac{N \beta_0}{d_{kj}^2} \cdot \frac{N \beta_0}{d_{jj}^2}} \nonumber \\
		&= \frac{1}{N^2} \left| \sum\nolimits_{n=0}^{N-1} e^{-i n \pi \Delta_{kj}} \right|^2 \nonumber \\
		&\overset{(a)}{=} \frac{1}{N^2} \left| \frac{\sin\left( \frac{N \pi}{2} \Delta_{kj} \right)}{\sin\left( \frac{\pi}{2} \Delta_{kj} \right)} \right|^2,
	\end{align}
	where $(a)$ follows by applying the geometric series sum formula $\sum_{n=0}^{N-1} q^n = \frac{1-q^N}{1-q}$ and Euler's identity $\sin(x) =~\frac{e^{ix} - e^{-ix}}{2i}$.
\end{IEEEproof}

Applying the results of Lemma \ref{lemma_upper_bound} and Lemma \ref{lemma_correlation}, we obtain the optimal rotation angles in the following theorem.
\begin{thm}\label{theorem_optimal_angle}
	For $K = 2$, the optimal rotation angles $\phi_j^*$ for problem \eqref{P1} are 
	\begin{align}\label{eq_optimal_angle}
		\phi_j^* = & \frac{\theta_{kj} + \theta_{jj}}{2}\nonumber \\
		& \pm \arccos \left[ \frac{m}{N \sin\left( \frac{\theta_{kj} - \theta_{jj}}{2} \right)} \right], \forall j \in \{1, 2\},
	\end{align}
	provided that $\left| \sin\left( \frac{\theta_{jj} - \theta_{kj}}{2} \right) \right| \ge \frac{|m|}{N}$, where $m \in \mathbb{Z}^+$.
\end{thm}
\begin{IEEEproof}
	Using the result in Lemma \ref{lemma_correlation}, $\sin(\theta_{kj} - \phi_j) - \sin(\theta_{jj} - \phi_j) = \frac{2m}{N}$. Applying the identity $\sin A - \sin B = 2\cos \left( \frac{A + B}{2} \right)\sin \left( \frac{A - B}{2} \right)$, this condition simplifies to
	\begin{align}\label{eq_simplified_condition}
		2\cos \left( \frac{\theta_{kj} + \theta_{jj}}{2} - \phi_j \right)\sin \left( \frac{\theta_{kj} - \theta_{jj}}{2} \right) &= \frac{2m}{N},
	\end{align}
	which leads to the closed-form optimal solution in \eqref{eq_optimal_angle}.
\end{IEEEproof}

	From Theorem \ref{theorem_optimal_angle}, $\left| \sin\left( \frac{\theta_{kj} - \theta_{jj}}{2} \right) \right| \ge \frac{1}{N}$ serves as the condition for achieving perfect orthogonality between the desired channel and interference channel. Physically, it specifies the minimum spatial angle difference required between the desired and interfering receivers to enable complete interference cancellation without loss of desired signal power. Furthermore, this condition can be easily satisfied in practical systems since the receiver is generally associated with its near transmitter.

	With the optimal rotation angles $\phi_j^*$ obtained in Theorem \ref{theorem_optimal_angle}, the optimal transmit beamforming at each transmitter reduces to the simple MRT as
	\begin{align}
		\mathbf{w}_j^* = \sqrt{P_{\max}} \frac{\mathbf{h}_{jj}(\phi_j^*)}{\|\mathbf{h}_{jj}(\phi_j^*)\|}, \quad \forall j \in \mathcal{J}.
	\end{align}
	
\begin{rem}\label{remark_MRT_optimality}
	In conventional MISO IFCs, treating interference as noise with Gaussian signaling is strictly suboptimal. In contrast, exploiting RAs provides additional degrees of freedom to optimize array rotation angles. This yields ideal spatial orthogonality between the desired channel and the interference channel, allowing MRT to perfectly cancel interference without losing desired signal power. In this case, Gaussian signaling with MRT beamforming becomes optimal from an information-theoretical viewpoint.
\end{rem}

\subsection{General Multi-receiver Scenario}
Next, we study the general case with $K$ transmitter-receiver pairs. We first provide an upper bound on the performance gap between the system performance under arbitrary rotation angles and the interference-free bound. Inspired by this gap analysis, we then minimize the interference leakage power by optimizing the rotation angles of antenna arrays. Finally, the optimal transmit beamforming is obtained under the optimized rotation angles.

\subsubsection{Rate Gap Analysis}
To provide theoretical insights into the antenna orientation design, we first quantitatively analyze the achievable rate gap between the system performance under arbitrarily rotation angles and the ideal interference-free upper bound.
Let $R_k(\{\phi_j\})$ denote the achievable rate of receiver $k$ under arbitrary rotation angles $\{\phi_j\}$, and let $R_k^{\text{up}}$ denote the interference-free rate upper bound for receiver $k$. Accordingly, the resulting rate gap is $\Delta R_k =~R_k^{\text{up}} - R_k(\{\phi_j\})$. We provide an upper bound on $\Delta R_k$ in the following theorem.

\begin{thm}\label{theorem_rate_gap}
	For any receiver $k$, the rate gap $\Delta R_k$ is upper-bounded by
	\begin{align}\label{eq_rate_gap_bound}
		\Delta R_k \le \log_2 \left( 1 + \frac{P_{\max} N \beta_0}{\sigma_k^2} \sum\nolimits_{j \neq k} \frac{\rho_{kj}^2(\phi_j)}{d_{kj}^2} \right),
	\end{align}
	where $\rho_{kj}^2(\phi_j) = \frac{|\mathbf{h}_{kj}^H(\phi_j) \mathbf{h}_{jj}(\phi_j)|^2}{\|\mathbf{h}_{kj}(\phi_j)\|^2 \|\mathbf{h}_{jj}(\phi_j)\|^2}$.
\end{thm}

\begin{IEEEproof}
	For a given receiver $k$, the ideal interference-free rate upper bound $R_k^{\text{up}}$ is given by
	\begin{equation}
		R_k^{\text{up}} = \log_2 \left( 1 + \frac{P_{\max} N \beta_0}{d_{kk}^2 \sigma_k^2} \right).
	\end{equation}
	
	For arbitrary rotation angles $\{\phi_j\}$, the achieved rate $R_k(\{\phi_j\})$ admits a lower bound when MRT beamforming is adopted. Specifically, substituting the MRT beamforming vector $\mathbf{w}_j^{\text{MRT}} = \sqrt{P_{\max}} \frac{\mathbf{h}_{jj}(\phi_j)}{\|\mathbf{h}_{jj}(\phi_j)\|}$ into the interference term yields the instantaneous interference power leaked from transmitter $j$ to receiver $k$ as
	\begin{align}
		I_{kj} &= |\mathbf{h}_{kj}^H(\phi_j) \mathbf{w}_j^{\text{MRT}}|^2 \nonumber \\
		&= P_{\max} \frac{|\mathbf{h}_{kj}^H(\phi_j) \mathbf{h}_{jj}(\phi_j)|^2}{\|\mathbf{h}_{jj}(\phi_j)\|^2} = P_{\max} \frac{N \beta_0}{d_{kj}^2} \rho_{kj}^2(\phi_j).
	\end{align}
	Let $I_k = \sum\nolimits_{j \neq k} I_{kj}$ denote the total interference leakage under MRT beamforming. Then, a lower bound on the achievable rate, denoted as $R_k^{\text{MRT}}(\{\phi_j\})$, is given by
	\begin{equation}
		R_k^{\text{MRT}}(\{\phi_j\}) = \log_2 \left( 1 + \frac{S_k \sigma_k^2}{I_k + \sigma_k^2} \right),
	\end{equation}
	where $S_k = \frac{P_{\max} N \beta_0}{d_{kk}^2 \sigma_k^2}$.
	Using this lower bound, the rate gap $\Delta R_k = R_k^{\text{up}} - R_k(\{\phi_j\})$ can be upper-bounded as
	\begin{align}\label{eq_app_gap_derive}
		\Delta R_k &\le R_k^{\text{up}} - R_k^{\text{MRT}}(\{\phi_j\}) \nonumber \\
		&= \log_2(1 + S_k) - \log_2 \left( 1 + \frac{S_k \sigma_k^2}{I_k + \sigma_k^2} \right) \nonumber \\
		&= \log_2 \left( \frac{\sigma_k^2 + S_k \sigma_k^2}{\sigma_k^2 + S_k \sigma_k^2 + I_k} \cdot \frac{I_k + \sigma_k^2}{\sigma_k^2} \right).
	\end{align}
	Since $I_k \ge 0$, the first term in the product is no greater than 1. Therefore, we have 
	\begin{equation}
		\Delta R_k \le \log_2 \left( \frac{I_k + \sigma_k^2}{\sigma_k^2} \right).
	\end{equation}
	Substituting the expression for $I_k$ into this inequality completes the proof.
\end{IEEEproof}

\subsubsection{Optimization of Rotation Angle}
Theorem \ref{theorem_rate_gap} directly reveals the optimization criterion for the antenna rotation angles. To closely approach the information-theoretic rate upper bound, the rotation angle $\phi_j$ at transmitter $j$ should be designed to minimize the total residual interference leakage across the network. Therefore, we define the total interference leakage caused by transmitter $j$ to non-associated receivers as
\begin{align}\label{eq_total_leakage_obj}
	I_j^{\text{total}}(\phi_j) = \sum\nolimits_{k \neq j} \frac{\rho_{kj}^2(\phi_j)}{d_{kj}^2}.
\end{align}
Accordingly, the optimization problem with respect to antenna rotation angles can be decoupled into $K$ independent single-variable sub-problems:
\begin{align}\label{eq_decoupled_phi}
	\min_{\phi_j} \quad & I_j^{\text{total}}(\phi_j) \nonumber \\
	\text{s.t.} \quad & 0 \le \phi_j < \pi, \quad \forall j \in \mathcal{J}.
\end{align}

Since this sub-problem involves only a single scalar variable $\phi_j$ over a bounded feasible region, each transmitter can construct this objective locally and perform a low-complexity fine one-dimensional (1D) grid search to obtain the optimal rotation angle $\phi_j^*$ in parallel.

\subsubsection{Optimization of transmit beamforming}
Once the optimized rotation angles $\{\phi_j^*\}_{j=1}^K$ are obtained, the antenna array orientations are fixed. At this stage, the effective MISO channel from transmitter $j$ to receiver $k$ reduces to a fixed vector $\mathbf{g}_{kj} = \mathbf{h}_{kj}(\phi_j^*)$.

Consequently, the original non-convex joint optimization problem reduces to a conventional MISO IFC beamforming problem with fixed antenna configuration. For a given sum-rate target $R>0$ and rate-profile weights $\zeta_k$, the transmit beamforming feasibility problem is formulated as
\begin{subequations}\label{eq_feasibility_problem}
	\begin{align}
		\text{Find: } & \{\mathbf{w}_j\}_{j=1}^K \nonumber \\
		\text{s.t. } \; & \frac{|\mathbf{g}_{kk}^H \mathbf{w}_k|^2}{2^{\zeta_k R} - 1} \ge \sum\nolimits_{j \neq k} |\mathbf{g}_{kj}^H \mathbf{w}_j|^2 + \sigma_k^2, \quad \forall k \in \mathcal{K}, \label{eq_feasibility_rate}\\
		& \|\mathbf{w}_j\|^2 \le P_{\max}, \quad \forall j \in \mathcal{J}.
	\end{align}
\end{subequations}
Note that without loss of optimality, a common phase shift can be applied to $\mathbf{w}_k$ such that $\mathbf{g}_{kk}^H \mathbf{w}_k$ is a real value for all $k$. Consequently, the non-convex rate constraints in \eqref{eq_feasibility_rate} can be equivalently recast as second-order cone (SOC) constraints:
\begin{equation}\label{eq_socp_constraint}
	\sqrt{\sum_{j \neq k} |\mathbf{g}_{kj}^H \mathbf{w}_j|^2 + \sigma_k^2} \le \frac{\operatorname{Re}(\mathbf{g}_{kk}^H \mathbf{w}_k)}{\sqrt{2^{\zeta_k R} - 1}}, \quad \forall k \in \mathcal{K}.
\end{equation}
By replacing \eqref{eq_feasibility_rate} with \eqref{eq_socp_constraint}, the feasibility problem becomes a standard SOCP, which is convex and can be efficiently solved by standard convex optimization solvers, such as CVX.

To find the optimal sum-rate $R^*$, we combine the SOCP solver with a bisection search over $R$. Solving this feasibility problem yields the optimal beamforming vectors $\{\mathbf{w}_j^*\}_{j=1}^K$ and the maximum sum-rate $R^*$.
The proposed algorithm consists of two steps: a 1D search for the rotation angles and SOCP-based beamforming optimization. For the first step, it requires solving $K$ parallel optimization problems, resulting in a complexity of $\mathcal{O}(KG)$ for all $K$ transmitters, where $G$ is the number of search grids. The complexity of solving \eqref{eq_feasibility_problem} is $\mathcal{O}(K^{3.5}N^{3.5})$. Therefore, the overall complexity of the proposed algorithm is $\mathcal{O}(KG + I K^{3.5}N^{3.5})$, where $I$ denotes the number of bisection iterations.

\section{Numerical Results}
In this section, we provide numerical results to validate the performance of the proposed design in the RA-assisted MISO IFC. The system parameters are set as follows: $f_c = 10$ GHz, $N = 4$, $\beta_0 = 10^{-4}$, and $\sigma_k^2 = -70$ dBm.
For comparison, we consider the following schemes: 
1) FPA with MRT: the conventional FPA system without rotation (i.e., $\phi_j = 0$) employs MRT beamforming; 
2) RA with MRT: the proposed RA system optimizes the rotation angles and adopts MRT beamforming; 
3) FPA with ZF: the FPA system employs zero-forcing (ZF) beamforming; 
4) RA with ZF: the RA system utilizes ZF beamforming; 
5) FPA with optimal beamforming: the FPA system obtains the optimal transmit beamforming via SOCP; 
6) RA with optimal beamforming: the RA system obtains the optimal transmit beamforming via SOCP; and 
7) Interference-free bound: the achievable rate upper bound assuming no interference from other transmitters.

\subsection{Two-receiver Scenario}
In the two-receiver scenario, transmitters Tx1 and Tx2 are located at $\left( {0,{\rm{ }}0} \right)$ and $\left( {100,{\rm{ }}0} \right)$ m, respectively, while receivers Rx1 and Rx2 are located at $\left( {0,{\rm{ }}120} \right)$ and $\left( {120,{\rm{ }}150} \right)$ m, respectively. 

\begin{figure*}[!t]
	\centering
	\begin{minipage}[t]{0.32\textwidth}
		\centering
		\includegraphics[width=\textwidth]{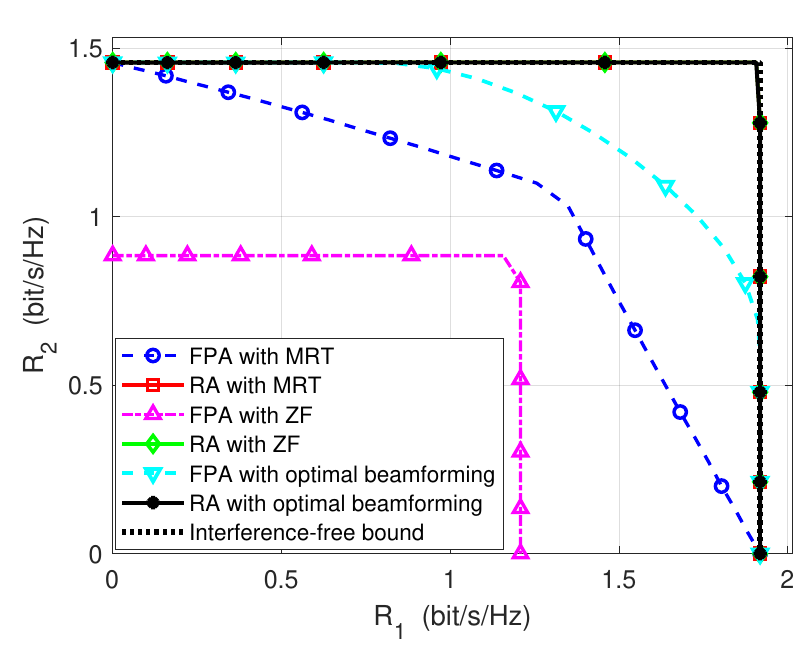}
		\caption{Achievable rate boundary under the low-SNR regime (${P_{\max }} = 10$ dBm).}
		\label{fig_low_snr}
	\end{minipage}
	\hfill
	\begin{minipage}[t]{0.32\textwidth}
		\centering
		\includegraphics[width=\textwidth]{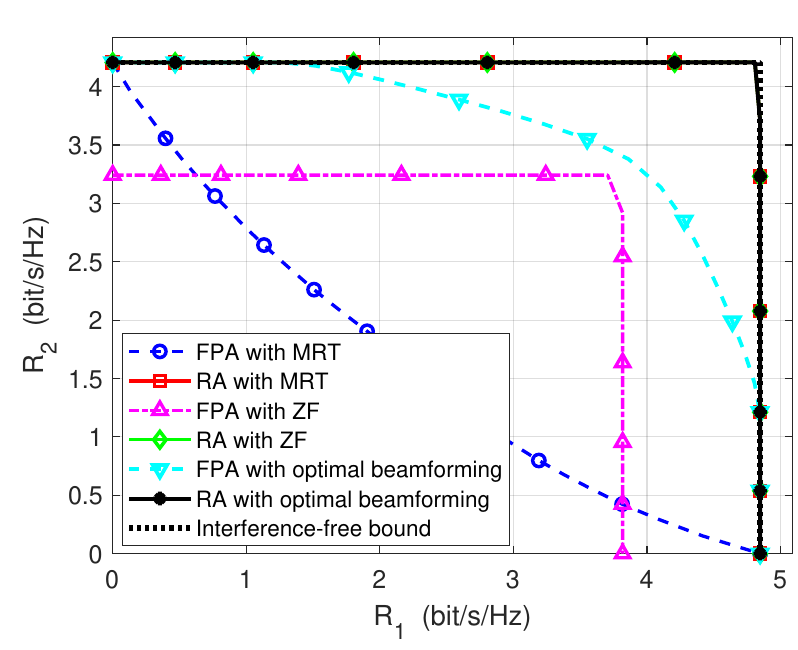}
		\caption{Achievable rate boundary under the high-SNR regime (${P_{\max }} = 20$ dBm).}
		\label{fig_high_snr}
	\end{minipage}
	\hfill
	\begin{minipage}[t]{0.32\textwidth}
		\centering
		\includegraphics[width=\textwidth]{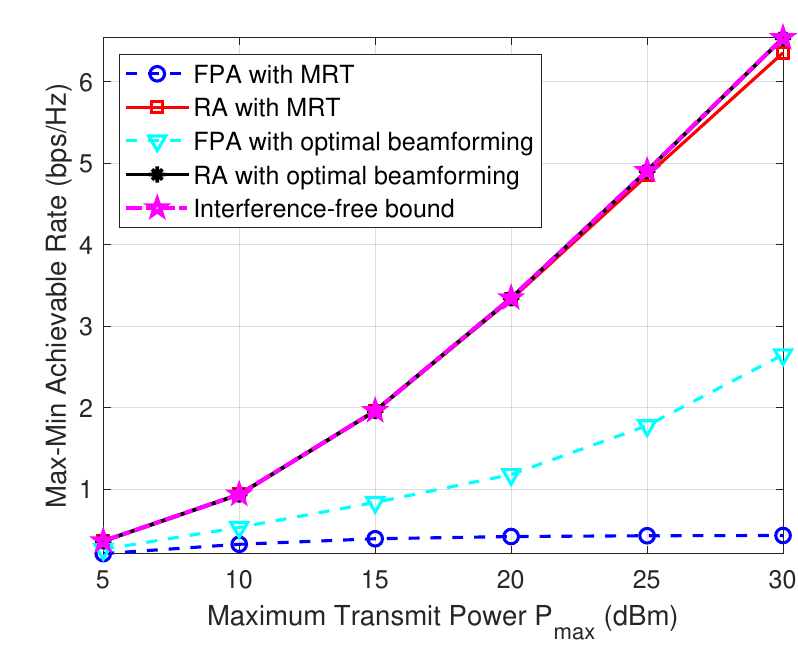}
		\caption{Max-min achievable rate versus the maximum transmit power $P_{\max}$ with $K = 4$.}
		\label{fig_multi_user}
	\end{minipage}
\end{figure*}

Fig. \ref{fig_low_snr} and Fig. \ref{fig_high_snr} illustrate the achievable rate regions in low-SNR (${P_{\max }} = 10$ dBm) and high-SNR (${P_{\max }} = 20$ dBm) regimes, respectively. For conventional FPA systems with $\phi_1 = \phi_2 = 0$, MRT performs slightly better in the low-SNR regime, whereas ZF beamforming yields a larger rate region in the high-SNR regime. By optimizing the RA rotation angles, the achievable rate region is significantly enlarged. Simulation results demonstrate that under both low and high SNR conditions, the RA-optimized MRT and ZF strategies achieve performance nearly identical to the interference-free upper bound. This validates that the additional spatial degree of freedom provided by RAs can not only significantly improve system performance, but also simplify transceiver design by enabling simple beamforming schemes with negligible performance loss.

\subsection{Multi-receiver Scenario}
Next, we evaluate the system performance of the general multi-receiver case by setting $K = 4$. The transmitters are located at $\left( {0,{\rm{ }}0} \right)$, $\left( {0,{\rm{ }}100} \right)$, $\left( {0,{\rm{ }}-1200} \right)$, and $\left( {0,{\rm{ }}-1300} \right)$ m, and the receivers are placed at $\left( {60,{\rm{ }}200} \right)$, $\left( {-60,{\rm{ }}200} \right)$, $\left( {60,{\rm{ }}-1200} \right)$, and $\left( {-60,{\rm{ }}-1300} \right)$ m.

Fig. \ref{fig_multi_user} plots the minimum achievable rate among all users versus the maximum transmit power $P_{\max}$. In this topology, due to the high spatial correlation among the channel vectors, the conventional FPA system experiences a severe performance loss compared to the interference-free upper bound. In contrast, by exploiting the new optimization dimension offered by RAs, the proposed design can closely approach the interference-free upper bound. As shown, the optimized RA combined with simple MRT beamforming achieves performance comparable to that of the RA with optimal beamforming, and substantially outperforms the conventional FPA network. This result indicates that deploying RAs can provide superior communication performance with a simplified transceiver design.

\section{Conclusion}
This paper investigated the potential of RAs for MISO IFC. We exploited the additional spatial degree of freedom provided by RA to create favorable wireless channel conditions, thereby enlarging the achievable rate region of the interference network. To this end, we jointly optimized the RA rotation angles and transmit beamforming to characterize the system achievable rate region. Numerical results demonstrated that the proposed RA-enabled MISO interference network significantly outperforms its conventional FPA counterpart. Furthermore, this highlights the significance of RAs in simplifying the design of transceiver transmission schemes.

\bibliographystyle{IEEEtran}

\end{document}